\documentclass[preprint , %superscriptaddress
]{revtex4-1}
\usepackage{graphicx}
\usepackage{amsfonts}
\usepackage{amssymb}
\usepackage{amsthm}
\usepackage{array}
\usepackage{amsmath}
\usepackage{verbatim} 
\usepackage{hyperref}
\usepackage{color}
\usepackage{bbold}
\usepackage{epstopdf}
\usepackage{mathtools}
\usepackage{enumerate}
\usepackage{subcaption}
\usepackage{color}
\usepackage{qcircuit}

\newtheorem{proposition}{Proposition}
\newtheorem*{proposition*}{Proposition}

\newtheorem{theorem}{Theorem}
\newtheorem*{theorem*}{Theorem}

\newtheorem*{corollary*}{Corollary}

\newtheorem{lemma}{Lemma}

\newcommand{\ket}[1]{\left\vert#1\right\rangle}
\newcommand{\bra}[1]{\left\langle#1\right\vert}

\newcommand{\cnot}[1]{U^\mathrm{CNOT}_{ #1 } }

\def\bra#1{\langle #1|}
\def\ket#1{\left|#1 \right>}

\def\Tr{\mbox{Tr}}

\begin{document}
\title{Quantum error correction assisted quantum metrology without entanglement}
\author{Kok Chuan Tan}
\email{bbtankc@gmail.com}
\author{S. Omkar}
\author{Hyunseok Jeong}
\email{jeongh@snu.ac.kr}
\affiliation{Center for Macroscopic Quantum Control \& Institute of Applied Physics, Department of Physics and Astronomy, Seoul National University, Seoul, 08826, Korea}

\begin{abstract}
In this article we study the role that quantum resources play in quantum error correction assisted quantum metrology (QECQM) schemes. We show that there exist classes of such problems where entanglement is not necessary to retrieve noise free evolution and Heisenberg scaling in the long time limit. Over short time scales, noise free evolution is also possible even without any form of quantum correlations. In particular, for qubit probes, we show that whenever noise free quantum metrology is possible via QECQM, entanglement free schemes over long time scales and correlation free schemes over short time scales are always possible.
\end{abstract}

\maketitle

\section{Introduction}

Quantum information science has gained prominence as an area of research in the recent decades. One of the key promises of the field is that the the quantum regime contains intrinsic advantages over classical theories that can be exploited for a variety of informational tasks. A promising area of study which has gained considerable attention recently is the application of quantum error correction techniques to enhance the precision of quantum metrology~\cite{Kessler2014, Arrad2014, Dur2014}. Quantum metrology concerns itself with the precise estimation of some unknown physical parameter, but the precision of such tools often require the preparation of nonclassical quantum states that are sensitive to decoherence effects~\cite{Huelga1997, Dobrzanski2009, Escher2011}. Quantum error correction thus offers the promise of enhancing precision by reducing the amount of noise acting on the system.

Another key concern in quantum information is the study of the differences between quantum and classical theories, leading to the development of a theory of quantum resources. Examples of quantum resources include entanglement~\cite{Horodecki2001} and quantum coherence~\cite{Streltsov2017}. Quantum entanglement is at present a well established quantum resource with many applications such as cryptography~\cite{Ekert1991}, teleportation~\cite{Bennett1991} and superdense coding~\cite{Bennett1992}. In comparison, the resource theory of quantum coherence is a recent theoretical development, with applications in topics as diverse as quantum macroscopicity~\cite{Yadin2015, Kwon2016}, quantum optics~\cite{Bagan2016, Tan2017} and quantum metrology~\cite{Tan2018}. It is worth nothing that entanglement and coherence are not entirely separate quantum resources, since entangled states generally contains coherence, though the converse is not necessarily true\cite{Streltsov2015, Tan2016, Tan2018-2}.

In this article, we will examine the problem of the quantum resources that are necessary for quantum error correction protocols to succeed while simultaneously allowing for quantum enhanced metrology~\cite{Sekatski2017, Zhou2018}. Interestingly, we find that there exist regimes where this can occur without the presence of quantum entanglement, thus requiring us to invoke more general notions of nonclassicality such as quantum discord~\cite{Ollivier2001,Henderson2001} in order to account for the success such protocols. This joins a list of known applications for quantum discord in quantum information
~\cite{Cavalcanti2011, Datta2008, Chuan2012, Davic2012, Bobby2014}. Quantum discord was also considered previously in various other specialized metrological scenarios
~\cite{Modi2011, Cable2016, Girolami2014,  Braun2018}.  

We also show that in the extremal case of short interaction times, product states containing zero quantum correlation, but nonzero quantum coherence is sometimes sufficient to generate nontrivial Fisher information in a noise free manner. For qubit probes in particular, we prove that whenever quantum error correction assisted protocols are possible, then an entanglement free protocol over long time scales, or a quantum correlation free protocol over short time scales is also possible.

\section{Preliminaries}
Here, we review some basic notions concerning nonclassical quantum states that will be used in the paper. A more detailed description of quantum metrology, and the role of quantum error correction in metrology, will be provided in the next section.

First, we define the notion of coherence. Let $\rho$ be the density matrix of a quantum state. Then for a fixed basis $\{ \ket{i} \}$, if $\rho$ is not diagonal with respect to this basis, then we say that the state is coherent, or that the state contains coherence. 

Second, a pure, bipartite quantum state of the form $\ket{\psi}_1\bra{\psi}\otimes \ket{\phi}_2\bra{\phi}$ is referred to as a product state. A quantum density matrix $\rho$ that is expressible as a convex sum of product states $\rho = \sum_i p_i \ket{\psi_i}_1\bra{\psi_i}\otimes \ket{\phi_i}_2\bra{\phi_i}$ is called a separable state. Furthermore, if a state $\rho$ is not separable, then we say that the state is entangled.

We now introduce some notations. We will denote the canonical Pauli matrices on the $m$th qubit as $X_m$, $Y_m$ and $Z_m$ respectively. The computational basis refers to the basis $\{ \ket{0}, \ket{1} \}$, from which we can define the states $\ket{+} \coloneqq  \frac{1}{\sqrt{2}}(\ket{0}+ \ket{1})$ and $\ket{-} \coloneqq  \frac{1}{\sqrt{2}}(\ket{0} - \ket{1})$. The unitary performing a CNOT operation between the $m$th and $n$th qubits is denoted $U^{\mathrm{CNOT}}_{mn}$ where the first subindex $m$ is the control qubit, i.e. $U^{\mathrm{CNOT}}_{mn}\ket{0}_m \ket{\psi}_n = \ket{0}_m \ket{\psi}_n$ and $U^{\mathrm{CNOT}}_{mn}\ket{1}_m \ket{\psi}_n = \ket{1}_m X_n\ket{\psi}_n$.

\section{Error Correction in the Sequential Scheme for quantum metrology}

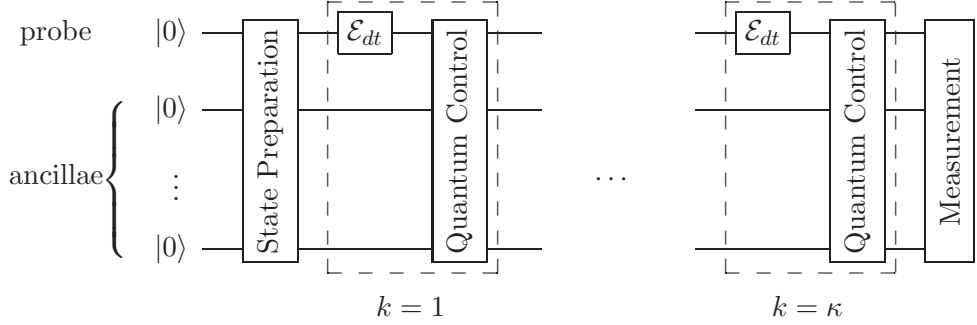
\begin{figure}[t]
\mbox{
\Qcircuit @C=1.4em @R=1.5em {
&\lstick{\mathrm{probe\;\;\;\;\;\;\;}}&\lstick{\ket{0}} & \multigate{3}{\rotatebox{90}{State Preparation}} & \gate{\mathcal{E}_{dt}} 
& \multigate{3}{\rotatebox{90}{Quantum Control}} &\qw & \nghost{} & \nghost{} &  \gate{\mathcal{E}_{dt}} & \multigate{3}{\rotatebox{90}{Quantum Control}} & \multigate{3}{\rotatebox{90}{Measurement}}\\
&&\lstick{\ket{0}} & \ghost{\rotatebox{90}{State Preparation}} & \qw& \ghost{\rotatebox{90}{Quantum Control}}& \qw &\nghost{} & \nghost{} & \qw & \ghost{\rotatebox{90}{Quantum Control}} & \ghost{\rotatebox{90}{Measurement}}\\
&&\lstick{\vdots\;} & \nghost{\rotatebox{90}{State Preparation}} &  \nghost{}& \nghost{\rotatebox{90}{Quantum Control}}& \nghost{} & \cdots & \nghost{} & \nghost{} & \nghost{\rotatebox{90}{Quantum Control}} & \nghost{\rotatebox{90}{Measurement}}\\
&&\lstick{\ket{0}} & \ghost{\rotatebox{90}{State Preparation}} & \qw & \ghost{\rotatebox{90}{Quantum Control}}& \qw & \nghost{} & \nghost{} & \qw & \ghost{\rotatebox{90}{Quantum Control}} & \ghost{\rotatebox{90}{Measurement}}
\inputgroupv{2}{4}{0.5em}{2.2em}{\mbox{ancillae\;\;\;\;}} 
\gategroup{1}{5}{4}{6}{0.7em}{--}
\gategroup{1}{10}{4}{11}{0.7em}{--}\\
&&&&&\lstick{k=1}&&&&&\lstick{k=\kappa}&&
}
}
\caption{The sequential scheme of quantum metrology. After state preparation, the probe is allowed to evolve for a short period $dt$. This evolution is represented by the quantum map $\mathcal{E}_{dt}$ which includes the interaction with some signal Hamiltonian, as well as contributions from noise. The ancillae are assumed to be perfectly shielded and noise free. This is followed by an instantaneous quantum control operation on probe+ancillae. The process is repeated for a total of $\kappa$ rounds, at the end of which a measurement is performed. The total interrogation time is $t=\kappa dt$.  
}
\label{fig::sequentialScheme}
\end{figure}

We will primarily consider error correcting strategies within the framework of the sequential scheme for quantum metrology. In the sequential scheme, an experimenter has access to a probe which can be initialized into any quantum state $\ket{\psi}$. This probe is subject to a Hamiltonian interaction of the form $H = \theta G$ which encodes a signal onto the probe. $G$ shall be referred to as the generator, to distinguish it from $H$. The experimenter will also have access to any number of noiseless ancillary particles, and the ability to perform accurate and fast quantum gates on the total probe-ancillae system. Notably, it is assumed that only the probe state interacts with the Hamiltonian and is affected by environmental noise. The quantum gates themselves are presumed to be ideal and instantaneous, while the ancillae are noiseless and perfectly shielded.

The objective of the experimenter is to obtain a measurement of the unknown quantity $\theta$ given the generator $G$. The parameter $\theta$ describes the strength of the interaction between the probe and the Hamiltonian $H$. In a noiseless scenario, the ultimate precision of this measurement is given by the Cram{\'e}r-Rao bound~\cite{Helstrom1976, Braunstein1994} $\delta \theta \geq \frac{1}{\sqrt{\nu F(\ket{\psi}, G)}}$, where $\nu$ is the number of times the experiment is repeated and $ F(\ket{\psi}, G)$ is the Fisher information quantity. In the quantum regime, it is possible for the measurement precision to achieve a scaling of $\sim\frac{1}{t}$. This is known as Heisenberg scaling (HS). However, quantum noise can diminish this scaling factor to $\sim\frac{1}{\sqrt{t}}$, which is known as the standard quantum limit (SQL). 

In this article, we will assume the noise is Markovian and the probe evolves according to the time homogeneous Lindblad equation. In its diagonal form, this is described by the Lindlad master equation~\cite{Lindblad1976, Nielsen2010} $$\frac{d\rho}{dt} = -i[H,\rho]+\sum_k (L_k\rho L^\dag_k - \frac{1}{2} \{ L_k^\dag L_k, \rho \}),$$ where $L_k$ describes the noisy part of the evolution and are called Linblad operators or jump operators.

The goal of the experimenter is to combat the effects of noise using the tools at his disposal. Since the experimenter has access to fast quantum gates, the experimenter may split the total interaction time $t$ into a total of $\kappa = \frac{t}{dt}$ rounds each of lasting a sufficiently short period of time $dt$, with $k= 1, \ldots, \kappa$ denoting the $k$th round. Within each of these rounds, the experimenter can implement a quantum error correction (QEC) scheme via fast quantum gates between the  ancillae and the probe, with the goal of (i) removing the noise component of the time evolution, and (ii) retrieving nontrivial noiseless evolution. In order to achieve nontrivial, noiseless evolution, the experimenter must carefully choose an error correction procedure that is able to correct the errors represented by the Lindblad operators $L_k$ whilst preserving the HS scaling. Figure~\ref{fig::sequentialScheme} illustrates the process. 

Such a strategy involves the use of full and fast quantum control (FFQC), and is sometimes referred to as FFQC assisted metrology. However, FFQC assisted metrology is in fact the most general framework for quantum error correction assisted quantum metrology considered thus far, with other scenarios existing as special cases of such strategies~\cite{Sekatski2017}. As such, in this article we will collectively refer to such strategies as quantum error correction assisted quantum metrology (QECQM) schemes. In~\cite{Zhou2018}, it was shown that as long as $H$ is not representable as a linear sum of the operators $\openone$, $L_k$, $L^\dag_k$ and $L^\dag_jL_k$ then QECQM is always possible. The set of all possible linear sums of such operators is known as the Lindblad span, which we denote by $\mathcal{S}$. Conversely, when this condition is not satisfied, then the experimenter can at best achieve SQL scaling, regardless of the strategy he employs.

\section{Noiseless evolution over short time scales}

In this section, we demonstrate that there exist Lindblad operators $L_k$, and corresponding nontrivial generators $G$ such that (i) HS is achieved over total interaction time $t$, and (ii) the probe-ancilla system is uncorrelated during the round $k=1$. Recalling that each round lasts for a period of $dt$, this suggests that over short time scales, it is possible for QECQM to be successful without quantum correlations.  

This is summarized by the following proposition:

\begin{proposition} \label{prop::rank2}
Let $\mathcal{S} = \mathrm{span}\{ \openone, L_k, L^\dag_k, L^\dag_kL_j \}$ be the Lindblad span. For any generator $G$, we can always write the decomposition $G= G_\parallel + G_\perp$ where  $G_\parallel$ and $G_\perp$ are the parallel and perpendicular components of $G$ w.r.t. the Lindblad span $\mathcal{S}$ and the operator inner product $\langle A,B \rangle \coloneqq \Tr(A^\dag B)$.

If $G_\perp$ is rank 2, then $G_\perp \propto \ket{c_0}\bra{c_0} - \ket{c_1}\bra{c_1} $ for some orthogonal vectors $\{ \ket{c_0}, \ket{c_1} \}$, and HS can always be achieved via QECQM for every timescale $t$. 

Furthermore, the initial probe-ancilla state can always be chosen such that it is a product state with no quantum correlations present during the round $k=1$. QECQM without quantum correlations is therefore possible over sufficiently a short timescale $dt$.
\end{proposition}

\begin{proof}
First, we observe that if $G_\perp$ is rank 2, then $G$ is not an element of $\mathcal{S}$ and thus not contained within the Lindblad span. As such, we know that there must exist some QECQM strategy that enables HS. Furthermore, $G_\perp$ is Hermitian since $G$ is Hermitian, and must be perpendicular to $\openone$, which is an element of $\mathcal{S}$. As such, we must have $\Tr(\openone G_\perp) = 0$. The only rank 2, Hermitian and traceless operator has the form $\lambda \ket{c_0}\bra{c_0} - \lambda \ket{c_1}\bra{c_1}$ for some orthogonal  $\{ \ket{c_0}, \ket{c_1} \}$ and $\lambda >0$, so $G_\perp \propto \ket{c_0}\bra{c_0} - \ket{c_1}\bra{c_1}$, which establishes the first part of the theorem.

For convenience, let us define $\ket{0} \coloneqq \ket{c_0}$ and $\ket{1} \coloneqq \ket{c_1}$ and $\ket{\pm} \coloneqq \frac{1}{\sqrt{2}}(\ket{0} \pm \ket{1})$. Let us choose the codespace  defined by $\{ \ket{0} \ket{+}, \ket{1} \ket{-} \}$ with corresponding projections $\Pi_C \coloneqq \ket{0}\bra{0}\otimes \ket{+} \bra{+} + \ket{1}\bra{1}\otimes \ket{-} \bra{-}$. We can verify that $\bra{0} \bra{+} O \otimes \openone \ket{1} \ket{-} = 0 $ and that $\bra{0} \bra{+} O \otimes \openone \ket{0} \ket{+} - \bra{1} \bra{-} O \otimes \openone \ket{1} \ket{-} = \Tr(G_\perp O)/\lambda$. Therefore, if $O$ is substituted with $L_k$ or $L_k^\dag L_j$, we have that $\Pi_C L_k \Pi_C = \mu_k \Pi_C$ and $\Pi_C L_k^\dag L_j \Pi_C = \mu_{k,j} \Pi_C$ which are exactly the error correction conditions so errors generated by $L_k$ or $L_k^\dag L_j$ are always correctable for any vector within this codespace. 

Furthermore, if we substitute $O$ with $G$, we get $\bra{0} \bra{+} G \otimes \openone \ket{0} \ket{+} - \bra{1} \bra{-} G \otimes \openone \ket{1} \ket{-} = \Tr(G_\perp ^2)/\lambda = 2 \lambda >0 $, so within the codespace defined by $\{ \ket{0} \ket{+}, \ket{1} \ket{-} \}$, the effective generator $\Pi_C G \Pi_C$ is nontrivial (i.e. it is not a constant). As such, within this codespace, the evolution is noiseless and nontrivial, and thus, HS within the QECQM framework can be achieved.

We observe that, by repeating similar arguments as above, the codespace defined by $\{ \ket{0} \ket{-}, \ket{1} \ket{+} \}$ and corresponding projectors $\Pi'_C \coloneqq \ket{0}\bra{0}\otimes \ket{-} \bra{-} + \ket{1}\bra{1}\otimes \ket{+} \bra{+}$ will similarly allow for errors generated by $L_k$ or $L_k^\dag L_j$ to be corrected and achieve HS.

Let $\{K_i\}$ be the Kraus operators~\cite{kraus} representing the error correcting map for $\Pi_C$. By definition, it must be able to correct errors of the form $L_k$ or $L_k^\dag L_j$, so $K_i O\otimes \openone \ket{0}\ket{+} \propto \ket{0}\ket{+}$ and $K_i O \otimes \openone \ket{1}\ket{-} \propto \ket{1}\ket{-}$ for every $i$, and $O$ can be any $L_k$ or $L_k^\dag L_j$. 

It is also clear that for $\Pi_C'$, the corresponding error correcting map is just $\{K_i'= (\openone \otimes Z) K_i (\openone \otimes Z) \}$ where $Z$ is the standard Pauli Z operator. Since 
$\{ \ket{0} \ket{+}, \ket{1} \ket{-} \}$ and $\{ \ket{0} \ket{-}, \ket{1} \ket{+} \}$, differs only by a phase flip on the ancilla, the projectors $\Pi_C$ and $\Pi_C'$ may be thought of as projections onto the even and odd parity subspaces. As such, if we define $K^{total}_i \coloneqq \Pi_C  K_i \Pi_C + \Pi'_C K_i' \Pi'_C$, every error generated by $L_k$ or $L_k^\dag L_j$ acting on the first qubit within the combined codespace $\{ \ket{0} \ket{+}, \ket{1} \ket{-} \} \cup \{ \ket{0} \ket{-}, \ket{1} \ket{+} \}$ can be corrected. The corresponding projector for this codespace is just $\Pi_C' \coloneqq (\ket{0}\bra{0}+ \ket{1}+\bra{1}) \otimes (\ket{+}\bra{+}+ \ket{-}+\bra{-})$.

We now compute the effective generator, and find that $\Pi_C' G \otimes \openone \Pi_C' = \lambda \ket{0}\bra{0} \otimes (\ket{+}\bra{+}+ \ket{-}+\bra{-}) - \lambda \ket{1}+\bra{1} \otimes (\ket{+}\bra{+}+ \ket{-}+\bra{-}) + \mathrm{constant} + \mathrm{off\;diagonal\;elements} $, so the effective generator is indeed nontrivial since the leading diagonal elements are not all equal. To achieve noiseless evolution and HS, we just need to choose from within the combined code space any vector that is not an eigenvector of $\Pi_C' G \otimes \openone \Pi_C'$. The equal superposition $\frac{1}{\sqrt{2}}(\ket{0}+\ket{1})\ket{+} = \ket{+} \ket{+}$ will suffice, as the only way this can possibly be an eigenvector of $\Pi_C' G \otimes \openone \Pi_C'$ is by having $\lambda = - \lambda = 0$, which is impossible since $\lambda > 0$. Therefore, this initial probe-ancilla state will generate nontrivial time evolution, and the state will be separable during round $k=1$. This completes the proof.

\end{proof}

Proposition~\ref{prop::rank2} is a technical result that establishes that whenever $G_\perp$ is rank 2, a product state is sufficient to successfully perform QECQM over time $dt$. The following lemma expands upon this observation by describing a class of Lindblad operators for which a generator of this type is guaranteed to exist. 

\begin{lemma} \label{lem::exist}
If the noisy evolution is described by a single Lindblad operator $L$ (i.e. the noise is rank 1), then there always exists some generator $G$ s.t. $\Tr(G L) = 0$, $G$ is traceless and rank 2.
\end{lemma}

\begin{proof}
First, we recall that the Lindblad operator is unique up to the addition of a constant. Therefore, we can always assume that $L$ is a traceless matrix. Any square, traceless matrix is unitarily similar to a zero diagonal matrix~\cite{Horn1985}. As such, we are guaranteed that there exists some orthonornal basis $\{ \ket{c_i} \}$ such that $\bra{c_i} L \ket{c_i} = 0$ for every $i$. Let us choose $G = \ket{c_0}\bra{c_0}-\ket{c_1}\bra{c_1}$. We see that $G$ is rank 2 and traceless. We can then directly verify that $\Tr(GL) = 0$ since $\bra{c_i} L \ket{c_i} = 0$, which proves the required result.
\end{proof}

Using Proposition~\ref{prop::rank2} and Lemma~\ref{lem::exist}, we now prove that for the qubit case, you can always choose the initial probe ancilla state to be a product state so long as HS is achievable.

\begin{theorem} \label{thm::qubit}
For a qubit probe subject to Markovian noise, if HS is achievable via QECQM, then you can always choose the initial probe-ancilla state such that it is a product state in round $k=1$. A product state is therefore sufficient to perform QECQM over time $dt$.

\end{theorem}

\begin{proof}
It is known that for a qubit probe, the only case where HS is achievable via QECQM is when the noisy evolution is described by a single Lindlad operator, and that QECQM is achievable only when $G \notin \mathrm{span}\{\openone, L\}$~\cite{Zhou2018}. If more Lindblad operators are necessary to describe the noise, then the Lindblad span will span the entire operator space of a qubit, and HS can never be achieved via QECQM since every Hamiltonian will be an element of the Lindblad span. From Proposition~\ref{prop::rank2} and Lemma~\ref{lem::exist}, we already know that \textit{some} generator $G$ will exist such that a product probe-ancilla state is possible in round $k=1$.

It remains to be shown that for every $G$ acting on a qubit that does not belong to the Lindblad span, a product probe-ancilla state is possible in round $k=1$. Since we can always assume that $L$ is traceless, it is always proportional to $\vec{a}\cdot\vec{\sigma}$ where $\vec{a}$ is a real 3 dimensional vector and $\vec{\sigma}$ is the usual vector of Pauli matrices~\cite{Zhou2018}. Similarly, since the the addition of a constant to the Hamiltonian does not change the time evolution, we can assume that the generator $G$ is also traceless, so $G = \vec{b} \cdot \vec{\sigma}$ for some real vector $\vec{b}$. Without any loss in generality, let us assume $\vec{a} = \hat{z}$. Then we can write $G = G_\parallel+G_\perp$ where $G_\parallel = b_z Z$ and $G_\perp = b_xX+b_yY$, where $X,Y,Z$ are the usual Pauli matrices. It is clear that $G_\perp$ is proportional to a Pauli matrix in the direction $(b_x,b_y, 0)$ and so must be rank 2. From Proposition~\ref{prop::rank2}, we see that a separable probe-ancilla state is possible in round $k=1$, which proves the required result.
\end{proof}

We can therefore conclude that over short time scales, quantum correlations are not a necessary prerequisite for QECQM. This is especially true for qubit probes, due to Theorem~\ref{thm::qubit}. We also note that the observations in Proposition~\ref{prop::rank2} and Lemma~\ref{lem::exist} are not necessarily limited to the qubit case, so such examples also exist in higher dimensions.  

\section{Example: Qubit probe with perpendicular noise.}

Here, we illustrate the case by examining a qubit probe with noise that is perpendicular to the Hamiltonian and generator. For simplicity, we will assume that $G=Z$ and $L=X$. We see that in this case, the generator $G$ is rank 2, and the noise is rank 1 so it can be described using only one Lindblad operator. From Theorem~\ref{thm::qubit}, we know for certain that we can always choose a product state as our initial probe-ancilla state. In this case the choice is especially simple. We will adopt the usual convention of letting the eigenvectors of the $Z$ Pauli matrix determine the computational basis. For the probe state, let us choose it to be  $\ket{+}_1$.  

In this case, the Lindblad master equation reads $$\frac{d\rho}{dt} = -i[Z,\rho]+ (X\rho X -  \rho ).$$ The substitution of $\rho = \ket{+}\bra{+}$ gives us $$d\rho = -i[Z,\rho]dt+ (X\rho X -  \rho )dt = -i[Z,\rho]dt, $$ which describes noiseless evolution over short timescales. As such, for a sufficiently short interaction time $dt$, no probe-ancilla correlations are necessary and the only quantum resource required is the local coherence of the probe state, which is necessary in order to generate nontrivial Fisher information. 

\section{Noiseless evolution over long timescales}

To achieve HS over every time scale, some form of quantum correlations is necessary.  However, we will demonstrate in this section that this does not have to be in the form of entanglement. 

In order to do that, we first describe a particular error correcting procedure. Consider the product state $\ket{+}_1\ket{0}_2$. We can perform a CNOT operation $\cnot{1,2}$ with qubit 1 acting as the control.  This leads to the maximally entangled state $\ket{\Phi^+}_{12} = \frac{1}{\sqrt{2}} (\ket{00}_{12} + \ket{11}_{12})$ after acting on the product state. Suppose we perform a bit flip operation $X_1$ on qubit 1 (the probe). This results in the state  $\ket{\psi^+}_{12} = \frac{1}{\sqrt{2}} (\ket{10}_{12} + \ket{01}_{12})$. The application of another identical CNOT operation will result in $\cnot{1,2}\ket{\psi_+}_{1,2} = \ket{+}_1\ket{1}_2$. For reasons that will be clear in the next paragraph, we also apply the operation $\cnot{2,1}$, which leaves the state unchanged. By observing the final state, we see that a Pauli $X_1$ operation is propagated from the probe to the ancilla. 

We repeat the argument for the product state $\ket{-}_1\ket{0}_2$. Applying the first CNOT leads to $\ket{\Phi^-} = \frac{1}{\sqrt{2}}(\ket{00}_{12}-\ket{11}_{12})$. After the Pauli error $X_1$, we get $- \ket{\Psi^-} = \frac{-1}{\sqrt{2}}(\ket{01}_{12}-\ket{10}_{12})$. Applying $\cnot{1,2}$, we get the $-\ket{-}_1\ket{1}_2$. Finally, we apply the operation $\cnot{2,1}$ and observe that this corrects the additional negative phase in front. The final state is $\ket{-}_1\ket{1}_2$. Again, we can see that a Pauli error $X_1$ on the probe is propagated to the ancilla. 

Let us choose the encoding procedure to be $E = \cnot{1,2}$, and the decoding procedure to be $D = \cnot{2,1}\cnot{1,2}$. With this, we are ensured that a Pauli $X_1$ error from the probe will always be propagated to the ancilla for any quantum superposition of the states $\ket{+}_1\ket{0}_2$ and $\ket{-}_1\ket{0}_2$.

Let us consider what happens when a Pauli $Z_1$ error occurs on the probe instead. Following the same encoding$\rightarrow$error$\rightarrow$decoding process described above, we find that $\ket{+}_1\ket{0}_2 \rightarrow \ket{-}_1\ket{0}_2$ and $\ket{-}_1\ket{0}_2 \rightarrow \ket{+}_1\ket{0}_2$. In summary, the error correcting procedure that was just described will always propagate a Pauli $X_1$ (bit flip) error to the ancilla, while Pauli $Z_1$ (phase flip) error is not propagated and remains on the probe qubit.

We now describe how the above protocol may be used to implement an entanglement free QECQM protocol. In the sequential scheme, we can initialize probe-ancilla in some state $\rho_1 \otimes\ket{0}_2\bra{0}$ before the start of every round, perform the encoding $E$, allow for free evolution of the probe, and end the round by performing the decoding $D$. If the evolution of the probe is described by a master equation of the form $$\frac{d\rho_1}{dt} = -i[Z_1,\rho_1]+ (X_1\rho_1 X_1 -  \rho_1 ),$$ we see that the noise is generated by the Pauli matrix $X_1$, while the signal is generated by the Pauli matric $Z_1$. This allows us to exploit the error propagation properties of the error correcting protocol. Due to the propagation of $X_1$ (bit flips) and the nonpropagation of $Z_1$ (phase flips) to the ancilla, the effective evolution of the probe is described by the noiseless evolution $$\frac{d\rho_1}{dt} = -i[Z_1,\rho_1],$$ while the evolution of the ancilla within each round is described entirely by the noisy part of the evolution $$\frac{d\rho_2}{dt} = (X_2\rho_2 X_2 -  \rho_2).$$ 

At the beginning of every round, we can repeat the error correcting process process by using a fresh ancilla initialized in the state $\ket{0}$.

Let us consider the probe state $\ket{\theta}_{1} \coloneqq \cos \theta \ket{0}_1 + \sin \theta \ket{1}_1$ and the ancilla $\ket{0}_2$. After performing the encoding $E$, the resulting state is the entangled pure state $\ket{\psi_\theta}_{1,2} \coloneqq \cos \theta \ket{0,0}_{1,2} + \sin \theta \ket{1,1}_{1,2}$. We now apply the Vidal-Tarrach Theorem~\cite{Vidal1999}, which states that every state of the form $\rho_\theta \coloneqq \frac{1}{1+s}\ket{\psi_\theta}_{1,2}\bra{\psi_\theta} + \frac{s}{4(1+s)} \openone$ is separable so long as $s \geq 2\sin(2\theta)$. For the choice $s=2$, $\rho_\theta$ will always separable for every $\theta$. We therefore see that over one round of the QECQM protocol, the preparation of the initial state $\frac{1}{1+s}\ket{\theta}_1\bra{\theta} \otimes \ket{0}_2\bra{0} + \frac{s}{4(1+s)} \openone_{1,2}$ where $s=2$ will ensure that the state is completely separable during the round. 

More generally, suppose the QECQM protocol is to be performed for  a total of $\kappa$ number of rounds and we use a fresh ancilla in the state $\ket{0}$ at the start of  every round of the error correction protocol. As a result, a probe-ancillae state of the form $\frac{1}{1+s}\ket{\theta}\bra{\theta} \otimes (\ket{0}\bra{0})^{\otimes \kappa} + \frac{s}{2^{\kappa+1}(1+s)} \openone$ will suffice to ensure that the total probe-ancilla state is never entangled throughout the entire process. The Fisher information is in this case $\mathcal{F}(\rho_1, G) = t^2 I $ where $$
I = 2 \sum_{i,j} \frac{(\lambda_i - \lambda_j)^2}{\lambda_i + \lambda_j}  |\bra{i}G\ket{j}|^2,$$ $\lambda_i$ and $\ket{i}$ are the eigenvalues and eigenvectors of $\rho_1$, and $\rho_1 = \frac{1}{1+s}\ket{\theta}_1\bra{\theta} + \frac{s}{4(1+s)} \openone_1$. This gives the required quadratic scaling for the Fisher information and demonstrates that entanglement is not a necessary resource to achieve noiseless time evolution and HS.

It is in fact always possible to implement an entanglement free protocol for every qubit probe that permits HS through QECQM. This is summarized in the following result:

\begin{theorem} \label{thm::qubitsep}
For a qubit probe subject to Markovian noise, if HS is achievable via QECQM, then you can always choose the probe-ancilla state and QEC protocol such that it is separable in every round $k=1, \ldots, \kappa = t/dt$, where $\kappa$ is the total number of rounds in the QECQM protocol, $t$ is the total interaction time, and $dt$ is the duration of each round.
\end{theorem}

\begin{proof}
This proof largely follows a modified version of the argument presented in Theorem~\ref{thm::qubit}, in conjunction with the error correcting protocol described above.

The only scenario where QECQM is possible for a qubit subject to Markovian noise is when $L \propto \vec{a}\cdot\vec{\sigma}$ where $\vec{a}$ is a real 3 dimensional vector. Similarly, since adding a constant to the generator $G$ does not change the Lindblad master equation, we can assume $G$ is traceless, and $G = \vec{b} \cdot \vec{\sigma}$ for some real vector $\vec{b}$. Without any loss in generality, let us assume $\vec{a} \propto \hat{x}$. Then we can write $G = G_\parallel+G_\perp$ where $G_\parallel = b_xX$ and $G_\perp = b_yY+b_zZ$, where $X,Y,Z$ are the usual Pauli matrices. It is clear that $G_\perp$ is proportional to a Pauli matrix in the direction $(0,b_y, b_z)$ , and since it is perpendicular to $\vec{a} \propto \hat{x}$, we can also assume without any loss in generality that $G_\perp \propto Z$. 

We then see that the contribution by $G_\parallel$ and $L$ are both proportional to $X$, while $G_\perp$ is proportional to $Z$. If we now apply the specific error correcting protocol described previously for perpendicular noise, we see that the contributions proportional to $X$ will propagate to the ancilla leaving $G_\perp$, which is proportional to $Z$, as the effective generator acting on the probe qubit. Since this protocol does not employ any entanglement, this shows that for qubit probes, an entanglement free protocol achieving HS always exists whenever a QECQM protocol achieving HS exists. 
\end{proof}

Theorem~\ref{thm::qubitsep} above demonstrates that in the qubit case, it is always possible to perform an entanglement free protocol provided we are allowed to input a mixed state for the probe. This noiseless evolution will extend to arbitrarily long time scales, given the assumptions of QMQEC. In this context, we can view Theorem~\ref{thm::qubit} as an extremal case that applies over sufficiently brief time scales, where not only is entanglement not a necessary prerequisite, no form of quantum correlations is necessary.

\section{Conclusion}

In this article, we considered the short time scale limit and find that noiseless time evolution (as well as the quadratic scaling of Fisher information over the interaction time), can be recovered over short periods of time without the presence of any form of quantum correlations. In this case, the only form of quantumness that is strictly necessary is quantum coherence, in order to generate nontrivial Fisher information with respect to the effective Hamiltonian driving the evolution of state. 

We also demonstrate that there exist nontrivial QECQM scenarios where the recovery of noiseless time evolution and HS in the long time limit can be achieved via separable states. This is sufficient for us to conclude that a successful QECQM protocol is not predicated on the existence of entanglement. In order to recover HS over long time scales however, some form of quantum correlations does appear to be necessary, so it is natural to suggest that successful QECQM protocols are in fact driven by some more generalized form of quantum correlations such as quantum discord.    

The fact that there exist time scales over which even quantum correlations are not necessary for noiseless evolution further complicates the issue of attributing a single quantum resource to a successful QECQM protocol. The simplest way to resolve this appears to be to treat the issue as 2 separate regimes. If one wishes to attribute a quantum resource to a successful QECQM protocol over \textit{long} interaction times in particular, then generalized quantum correlations appear to provide the answer. On the other hand, if one wishes to attribute a quantum resource to a successful QECQM protocol over \textit{every} time scale, then this necessitates a notion of quantumness even more general than quantum correlations. At present, the leading candidate for this is quantum coherence. We note that in general it is possible to generalize coherence to include notions of quantum correlations such as entanglement and discord.

It is also interesting to note that for qubit probes in particular, if QECQM protocols are feasible, then entanglement free protocols over long time scales, or quantum correlation free protocols over short time scales is always possible (see Theorems~\ref{thm::qubit} and \ref{thm::qubitsep}). At present, it remains an open question if this connection between the feasibility of QECQM and quantum correlations persist in higher dimensions systems. 

We hope that our results will spur further interest into the research of the role of quantum resources in quantum metrology.

\section{Acknowledgements}

This work was supported by the National Research Foundation of Korea (NRF) through a grant funded by the Korea government (MSIP) (Grant No. 2010-0018295) and by the Korea Institute of Science and Technology Institutional Program (Project No. 2E27800-18-P043). K.C. Tan was supported by Korea Research Fellowship Program through the National Research Foundation of Korea (NRF) funded by the Ministry of Science and ICT (Grant No. 2016H1D3A1938100).

\end{document}